%% file: main.tex
\documentclass[UKenglish]{eptcs}

\usepackage{prooftree}

\usepackage{mathptmx}
\usepackage{amssymb,amsthm}
\usepackage{url}
\usepackage{color}
\usepackage{upgreek}
\usepackage{stmaryrd}
\usepackage{mathtools}
\usepackage{etoolbox}
\usepackage{enumitem}
\usepackage{mathpartir}
\usepackage{soul}
\usepackage{wrapfig}
\setlist{itemsep=0pt}

\usepackage{tikz}
\usetikzlibrary{shapes}
\usetikzlibrary{calc}

\input{macros}

\providecommand{\event}{PLACES 2014}

\title{Session Type Isomorphisms}
\def\titlerunning{Session Type Isomorphisms}

\author{
  Mariangiola Dezani-Ciancaglini\institute{
  Universit\`a di Torino, Italy}  \and
  Luca Padovani\institute{
  Universit\`a di Torino, Italy}
  \and
  Jovanka Pantovic\institute{
 Univerzitet u Novom Sadu, Serbia
 }
}

\def\authorrunning{Dezani-Ciancaglini, Padovani, Pantovic}

\begin{document}
\maketitle

\begin{abstract}
There has been a considerable amount of work on retrieving functions
in function libraries using their type as search key. The availability
of rich component specifications, in the form of behavioral types,
enables similar queries where one can search a component library using
the behavioral type of a component as the search key.  Just like for
function libraries, however, component libraries will contain
components whose type differs from the searched one in the order of
messages or in the position of the branching points. Thus, it makes
sense to also look for those components whose type is different from,
but isomorphic to, the searched one.

In this article we give semantic and axiomatic characterizations of
isomorphic session types. The theory of session type isomorphisms
turns out to be subtle. In part this is due to the fact that it relies
on a non-standard notion of equivalence between processes. In
addition, we do not know whether the axiomatization is complete. It is
known that the isomorphisms for arrow, product and sum types are not
finitely axiomatisable, but it is not clear yet whether this negative
results holds also for the family of types we consider in this work.
\end{abstract}

\input{introduction}

\input{processes}

\input{types}

\input{axioms}
\input{adaptors}

\input{conclusions}

\bibliographystyle{eptcs}
\bibliography{main}


\end{document}

%% file: macros.tex
\definecolor{mymagenta}{rgb}{0.5,0,0.5}
\definecolor{mygreen}{rgb}{0,0.4,0}
\definecolor{myblue}{rgb}{0,0,0.6}
\definecolor{myred}{rgb}{0.4,0,0}
\definecolor{hlcolor}{rgb}{1,0.95,0}

\newcommand{\mkkeyword}[1]{\text{\texttt{\color{myblue}#1}}}
\newcommand{\mkconstant}[1]{\text{\texttt{\color{mygreen}#1}}}
\newcommand{\mktype}[1]{\text{\texttt{\color{myblue}#1}}}

\newcommand{\ie}{\emph{i.e.}}

\newcommand{\eqdef}{\stackrel{\text{\tiny def}}{=}}

\newcommand{\emphdef}[1]{\hl{\emph{#1}}}

\newcommand{\co}[1]{\overline{#1}}
\newcommand{\subst}[2]{\{#1/#2\}}
\newcommand{\set}[1]{\{#1\}}
\newcommand{\rulename}[1]{\text{\small[\textsc{#1}]}}
\newcommand{\eoe}{\hfill{$\blacksquare$}}

\newcommand{\Channel}{\mathsf{c}}

\newcommand{\ChannelL}{\mathtt{\color{myred}l}}
\newcommand{\ChannelR}{\mathtt{\color{myred}r}}

\newcommand{\cc}{\Channel}
\newcommand{\lc}{\ChannelL}
\newcommand{\rc}{\ChannelR}

\newcommand{\Selector}{\ell}
\newcommand{\selectL}{\mkkeyword{inl}}
\newcommand{\selectR}{\mkkeyword{inr}}

\newcommand{\Expr}{\mathsf{e}}

\newcommand{\var}{\varX}
\newcommand{\varX}{x}
\newcommand{\vartX}{x}
\newcommand{\varY}{y}
\newcommand{\vartY}{y}
\newcommand{\varZ}{z}

\newcommand{\Unit}{\mkconstant{()}}
\newcommand{\True}{\mkconstant{true}}
\newcommand{\False}{\mkconstant{false}}

\newcommand{\Value}{\mathsf{v}}

\newcommand{\x}{\varX}

\newcommand{\true}{\True}
\newcommand{\false}{\False}

\newcommand{\Context}{\mathcal{C}}
\newcommand{\ContextS}{\mathcal{E}}
\newcommand{\hole}{[~]}

\newcommand{\Process}{\ProcessP}
\newcommand{\ProcessP}{{P}}
\newcommand{\ProcessQ}{{Q}}
\newcommand{\ProcessR}{{R}}

\newcommand{\Adapter}{\AdapterA}
\newcommand{\AdapterA}{A}
\newcommand{\AdapterB}{B}

\newcommand{\nullp}{\mathbf{0}}
\newcommand{\inp}[3]{#1\text{\texttt{?}}(#2:#3)}

\newcommand{\outp}[2]{#1\text{\texttt{!}}\langle#2\rangle}
\newcommand{\selectp}[2]{#1 \triangleleft #2}
\newcommand{\casep}[3]{#1 \triangleright \{#2, #3\} }
\newcommand{\condp}[3]{\mkkeyword{if}~#1~\mkkeyword{then}~#2~\mkkeyword{else}~#3}

\newcommand{\pc}{\mathbin{{\rfloor}\hspace{-0.5ex}{\lceil}}}

\newcommand{\procout}[3]{\outp{#1}{#2}.#3}
\newcommand{\cond}[3]{\condp{#1}{#2}{#3}}
\newcommand{\PP}{\ProcessP}
\newcommand{\Q}{\ProcessQ}

\newcommand{\Sort}{\Type}
\newcommand{\Type}{\TypeT}
\newcommand{\TypeT}{t}
\newcommand{\TypeS}{s}

\newcommand{\tunit}{\mktype{unit}}
\newcommand{\tbool}{\mktype{bool}}
\newcommand{\tint}{\mktype{int}}

\renewcommand{\S}{\Sort}
\newcommand{\bool}{\tbool}
\renewcommand{\int}{\tint}

\newcommand{\unit}{\tunit}

\newcommand{\SessionType}{\SessionTypeT}
\newcommand{\SessionTypeT}{T}
\newcommand{\SessionTypeS}{S}

\newcommand{\End}{\mathtt{\color{myblue}end}}
\newcommand{\In}[1]{\text{\texttt{?}}#1}
\newcommand{\Out}[1]{\text{\texttt{!}}#1}

\newcommand{\T}{\SessionType}
\newcommand{\eend}{\End}
\newcommand{\tsel}[2]{#1\oplus #2}

\newcommand{\red}{\longrightarrow}
\newcommand{\ared}{\rightsquigarrow}
\newcommand{\aredLR}{\leftrightsquigarrow}

\newcommand{\nred}{\arrownot\red}
\newcommand{\eval}[2]{#1 \downarrow #2}
\newcommand{\eqp}{\approx}
\newcommand{\isom}{\cong}

\newcommand{\VEnv}{\upGamma}

\newcommand{\wte}[3]{#1 \vdash #2 : #3}
\newcommand{\wtpx}[6]{
  #1
  \vdash
  #2 \blacktriangleright \set{ #3 : #4, #5 : #6}
}
\newcommand{\wtp}[4]{\wtpx{#1}{#2}{\Channel}{#3}{\co\Channel}{#4}}
\newcommand{\wtplr}[4]{\wtpx{#1}{#2}{\ChannelL}{#3}{\ChannelR}{#4}}

\theoremstyle{definition}

\newtheorem{theorem}{Theorem}

\newtheorem{definition}{Definition}
\newtheorem{example}{Example}

\newcommand{\id}[1]{\mathsf{id}_{#1}}
\newcommand{\derlr}[4]{#1\vdash#2\blacktriangleright\set{\lc:#3,\rc:#4}}

\newcommand{\procL}[2]{\selectp #1\selectL.#2 }
\newcommand{\procR}[2]{\selectp #1 \selectR.#2}

\newcommand{\procbr}[3]{\casep {#1} {#2} {#3} }

\newcommand{\idA}[2]{\mathsf{id}_{#1}}



%% file: introduction.tex
\section{Introduction}
\label{sec:introduction}

We have all experienced, possibly during a travel abroad, using an ATM
that behaves differently from the ones we are familiar with. Although
the information requested for accomplishing a transaction is
essentially always the same -- the PIN, the amount of money we want to
withdraw, whether or not we want a receipt -- we may be prompted to
enter such information in an unexpected order, or we may be asked to
dismiss sudden popup windows containing informative messages --
``charges may apply'' -- or commercials. Subconsciously, we
\emph{adapt} our behavior so that it matches the one of the ATM we are
operating, and we can usually complete the transaction provided that
the expected and actual behaviors are \emph{sufficiently similar}. An
analogous problem arises during software development or execution,
when we need a component that exhibits some desired behavior while the
components we have at hand exhibit similar, but not exactly equal,
behaviors which could nonetheless be adapted to the one we want.  In
this article, we explore one particular way of realizing such
adaptation in the context of binary sessions, where the behavior of
components is specified as session types.

There are two key notions to be made precise in the previous
paragraph: first of all, we must clarify what it means for two
behaviors to be ``similar'' to the point that one can be adapted into
the other; second, as for the ``subconscious'' nature of adaptation,
we translate this into the ability to synthesize the adapter
automatically -- {\ie} without human intervention -- just by looking
at the differences between the required and actual behaviors of the
component.  Clearly we have to find a trade-off: the coarser the
similarity notion is the better, for this means widening the range of
components we can use; at the same time, it is reasonable to expect
that the more two components differ, the harder it gets to
automatically synthesize a sensible adapter between them.
The methodology we propose in this work is based on the theory of
\emph{type isomorphisms}~\cite{DiCosmo95}. Intuitively, two types
$\SessionTypeT$ and $\SessionTypeS$ are isomorphic if there exist two
adapters $\AdapterA : \SessionTypeT \to \SessionTypeS$ and $\AdapterB
: \SessionTypeS \to \SessionTypeT$ such that $\AdapterA$ transforms a
component of type (or, that behaves like) $\SessionTypeT$ into one of
type $\SessionTypeS$, and $\AdapterB$ does just the opposite. It is
required that these transformations must not \emph{lose any
  information}. This can be expressed saying that if we compose
$\AdapterA$ and $\AdapterB$ in any order they annihilate each other,
that is we obtain adapters $\AdapterA \pc \AdapterB : \SessionTypeT
\to \SessionTypeT$ and $\AdapterB \pc \AdapterA : \SessionTypeS \to
\SessionTypeS$ that are equivalent to the ``identity'' trasformations
on $\SessionTypeT$ and $\SessionTypeS$ respectively.

In the following we formalize these concepts:
we define syntax and semantics of processes as well as a notion of
process equivalence~(Section~\ref{sec:processes}).
Next, we introduce a type system for processes, the notion of session
type isomorphism, and show off samples of the transformations we can
capture in this framework (Section~\ref{sec:types}).
We conclude with a quick survery of related works and open
problems~(Section~\ref{sec:conclusions}).


%% file: processes.tex
\section{Processes}
\label{sec:processes}

We let $m$, $n$, $\dots$ range over integer numbers;
we let $\Channel$ range over the set $\set{\ChannelL, \ChannelR}$ of
\emphdef{channels} and $\Selector$ range over the set $\set{\selectL,
  \selectR}$ of \emphdef{selectors}. We define an involution
$\co{\,\cdot\,}$ over channels such that $\co\ChannelL = \ChannelR$.
We assume a set of \emphdef{basic values} $\Value$, $\dots$ and
\emphdef{basic types} $\TypeT$, $\TypeS$, $\dots$ that include the
unitary value $\Unit$ of type $\tunit$, the booleans $\True$ and
$\False$ of type $\tbool$, and the integer numbers of type $\tint$.
We write $\Value \in \Type$ meaning that $\Value$ has type $\Type$.
We use a countable set of \emphdef{variables} $\varX$, $\varY$, \dots;
\emphdef{expressions} $\Expr$, $\dots$ are either variables or values
or the equality $\Expr_1 = \Expr_2$ between two expressions.
Additional expression forms can be added without substantial issues.
\emphdef{Processes} are defined by the grammar
\[
\ProcessP ~~::=~~
\nullp
~~\mid~~ \inp\Channel\var\Type.\ProcessP
~~\mid~~ \outp\Channel\Expr.\ProcessP
~~\mid~~ \selectp\Channel\Selector.\ProcessP
~~\mid~~ \casep\Channel\ProcessP\ProcessQ
~~\mid~~ \condp\Expr\ProcessP\ProcessQ
~~\mid~~ \ProcessP \pc \ProcessQ
\]
which includes the terminated process $\nullp$, input
$\inp\Channel\var\Type.\Process$ and output
$\outp\Channel\Expr.\Process$ processes, as well as labeled-driven
selection $\selectp\Channel\Selector.\Process$ and branching
$\casep\Channel\ProcessP\ProcessQ$, the conditional process
$\condp\Expr\ProcessP\ProcessQ$, and parallel composition $\ProcessP
\pc \ProcessQ$.
The peculiarity of the calculus is that communication occurs only
between adjacent processes.  Such communication model is exemplified
by the diagram below which depicts the composition $\ProcessP \pc
\ProcessQ$.  Each process sends and receives messages through the 
channels $\ChannelL$ and~$\ChannelR$.
\begin{wrapfigure}[6]{r}{0.45\linewidth}
\vskip-1ex
\begin{tikzpicture}[semithick]
  \node (P) at (4,0) [cylinder, draw, minimum height=6em, shape aspect=.5] {$\mathstrut$};
  \node (Q) at (7,0) [cylinder, draw, minimum height=6em, shape aspect=.5] {$\mathstrut$};

  \draw[<-] ($(P.bottom) + (0,1ex)$) -- ++(-2em,0);
  \draw[->] ($(P.bottom) - (0,1ex)$) -- ++(-2em,0);

  \node at ($(P.center) + (0,2em)$) {$\ProcessP$};
  \node at ($(Q.center) + (0,2em)$) {$\ProcessQ$};
  \node at (5.5,2em) {$\pc$};
  \node at ($(P.before bottom) + (0,-1em)$) {$\ChannelL$};
  \node at ($(P.after top) + (0,-1em)$) {$\ChannelR$};
  \node at ($(Q.before bottom) + (0,-1em)$) {$\ChannelL$};
  \node at ($(Q.after top) + (0,-1em)$) {$\ChannelR$};

  \draw[->] ($(P.top) + (-1ex,1ex)$) -- ($(Q.bottom) + (0,1ex)$);
  \draw[<-] ($(P.top) + (-1ex,-1ex)$) -- ($(Q.bottom) + (0,-1ex)$);

  \draw[->] ($(Q.top) + (-1ex,1ex)$) -- ++(2em,0);
  \draw[<-] ($(Q.top) + (-1ex,-1ex)$) -- ++(2em,0);
\end{tikzpicture}
\end{wrapfigure}
Messages sent by $\ProcessP$ on
$\ChannelR$ are received by $\ProcessQ$ from $\ChannelL$, and messages
sent by $\ProcessQ$ on $\ChannelL$ are received by $\ProcessP$ from
$\ChannelR$.  Therefore, unlike more conventional parallel composition
operators, $\pc$ is associative but not symmetric in general.
Intuitively, $\ProcessP \pc \ProcessQ$ models a binary session where
$\ProcessP$ and $\ProcessQ$ are the processes accessing the two
endpoints of the session. By compositionality, we can also represent
more complex scenarios like $\ProcessP \pc \AdapterA \pc \ProcessQ$
where the interaction of the same two processes $\ProcessP$ and
$\ProcessQ$ is mediated by an adapter $\AdapterA$ that filters and/or
transforms the messages exchanged between $\ProcessP$ and
$\ProcessQ$. In turn, $\AdapterA$ may be the parallel composition of
several simpler adapters.

The operational semantics of processes is formalized as a reduction
relation closed by reduction contexts and a structural congruence
relation.
\emphdef{Reduction contexts} $\Context$ are defined by the grammar
\[
\Context ~~::=~~ \hole ~~\mid~~ \Context \pc \Process ~~\mid~~ \Process \pc \Context
\]
and, as usual, we write $\Context[\Process]$ for the process obtained
by replacing the hole in $\Context$ with $\Process$.

\begin{table}
\caption{\label{tab:reduction} Reduction relation.\strut}
\framebox[\textwidth]{
\begin{math}
\displaystyle
\begin{array}{c}
\inferrule[\rulename{r-comm 1}]{
  \eval\Expr\Value
  \\
  \Value \in \Type
}{
  \outp\ChannelR\Expr.\ProcessP
  \pc
  \inp\ChannelL\var\Type.\ProcessQ
  \red
  \ProcessP
  \pc
  \ProcessQ\subst\Value\var
}
\qquad
\inferrule[\rulename{r-comm 2}]{
  \eval\Expr\Value
  \\
  \Value \in \Type
}{
  \inp\ChannelR\var\Type.\ProcessP
  \pc
  \outp\ChannelL\Expr.\ProcessQ
  \red
  \ProcessP\subst\Value\var
  \pc
  \ProcessQ
}
\\\\
\inferrule[\rulename{r-choice 1}]{}{
  \selectp\ChannelR\Selector.\ProcessP
  \pc
  \casep\ChannelL{\ProcessQ_\selectL}{\ProcessQ_\selectR}
  \red
  \ProcessP \pc \ProcessQ_\Selector
}
\qquad
\inferrule[\rulename{r-choice 2}]{}{
  \casep\ChannelR{\ProcessP_\selectL}{\ProcessP_\selectR} \pc \selectp\ChannelL\Selector.\ProcessQ
  \red
  \ProcessP_\Selector \pc \ProcessQ
}
\\\\
\inferrule[\rulename{r-cond}]{
  \eval\Expr\Value
  \\
  \Value\in\tbool
}{
  \condp\Expr{\ProcessP_\True}{\ProcessP_\False} \red \ProcessP_\Value
}
\qquad
\inferrule[\rulename{r-context}]{
  \ProcessP \red \ProcessQ
}{
  \Context[\ProcessP] \red \Context[\ProcessQ]
}
\qquad
\inferrule[\rulename{r-struct}]{
  \ProcessP \equiv \ProcessP'
  \\
  \ProcessP' \red \ProcessQ'
  \\
  \ProcessQ' \equiv \ProcessQ
}{
  \ProcessP \red \ProcessQ
}
\end{array}
\end{math}
}
\end{table}

\emphdef{Structural congruence} is the least congruence defined by the
rules
\[
\nullp \pc \nullp \equiv \nullp
\qquad
\ProcessP \pc (\ProcessQ \pc \ProcessR) \equiv (\ProcessP \pc
\ProcessQ) \pc \ProcessR
\]
while \emphdef{reduction} is the least relation $\red$ defined by the
rules in Table~\ref{tab:reduction}. The rules are familiar and
therefore unremarkable. We assume a deterministic \emphdef{evaluation}
relation $\eval\Expr\Value$ expressing the fact that $\Value$ is the
value of $\Expr$. We write $\red^*$ for the reflexive, transitive
closure of $\red$ and $\ProcessP \nred$ if there is no $\ProcessQ$
such that $\ProcessP \red \ProcessQ$.
With these notions we can characterize the set of correct processes,
namely those that complete every interaction and eventually reduce to
$\nullp$:

\begin{definition}[correct process]
\label{def:correctness}
We say that a process $\Process$ is \emphdef{correct} if $\ProcessP
\red^* \ProcessQ \nred$ implies $\ProcessQ \equiv \nullp$.
\end{definition}

A key ingredient of our development is a notion of process equivalence
that relates two processes $\ProcessP$ and $\ProcessQ$ whenever they
can be completed by the same contexts $\Context$ to form a correct
process. Formally:

\begin{definition}[equivalence]
\label{def:eqp}
We say that two processes $\ProcessP$ and $\ProcessQ$ are
\emphdef{equivalent}, notation $\ProcessP \eqp \ProcessQ$, whenever
for every $\Context$ we have that $\Context[\ProcessP]$ is correct if
and only if $\Context[\ProcessQ]$ is correct.
\end{definition}

Note that the relation $\eqp$ differs from more conventional
equivalences between processes. In particular, 
 $\eqp$ is insensitive
 to the exact time when visible actions are made available on the two
 interfaces of a process. For example, we have
 \begin{equation}
 \label{eq:swap}
   \inp\ChannelL\varX\int.
   \outp\ChannelR\True.
   \inp\ChannelL\varY\unit
   \eqp
   \inp\ChannelL\varX\int.
   \inp\ChannelL\varY\unit.
   \outp\ChannelR\True
 \end{equation}
 despite the fact that the two processes perform visible actions in
 different orders.  Note that the processes in \eqref{eq:swap} are not
 (weakly) bisimilar.


%% file: types.tex
\section{Type System and Isomorphisms}
\label{sec:types}

\emphdef{Session types} $\SessionTypeT$, $\SessionTypeS$, $\dots$ are
defined by the grammar
\[
\SessionType ~~::=~~
\End
~~\mid~~ \In\Type.\SessionType
~~\mid~~ \Out\Type.\SessionType
~~\mid~~ \SessionTypeT + \SessionTypeS
~~\mid~~ \SessionTypeT \oplus \SessionTypeS
\]
and are fairly standard, except for branching $\SessionTypeT +
\SessionTypeS$ and selection $\SessionTypeT \oplus \SessionTypeS$
which are binary instead of $n$-ary operators, consistently with the
process language.
As usual, we denote by $\co\SessionType$ the \emphdef{dual} of
$\SessionType$, namely the session type obtained by swapping inputs
with outputs and selections with branches in $\SessionType$.

We let $\VEnv$ range over \emphdef{environments} which are finite maps
from variables to types of the form \[\var_1 : \Sort_1, \dots, \var_n :
\Sort_n.\]
The typing rules are given in Table~\ref{tab:typing_rules}.
Judgments have the form:
\begin{itemize}
\item  $\wte\VEnv\Expr\Sort$ stating that $\Expr$ is
well typed and has type $\Sort$ in the environment $\VEnv$ and
\item $\wtp\VEnv\Process{\SessionTypeT}{\SessionTypeS}$ stating that
$\Process$ is well typed in the environment $\VEnv$ and uses channel
$\Channel$ according to $\SessionTypeT$ and $\co\Channel$ according to
$\SessionTypeS$.
\end{itemize}

\begin{table}
\caption{Typing rules for expressions and processes.\strut}
\label{tab:typing_rules}
\framebox[\textwidth]{
\begin{math}
\displaystyle
\begin{array}{c}
\inferrule[\rulename{T-var}]{}{
  \wte{\VEnv, \var:\Sort}{\var}\Sort
}
\qquad
\inferrule[\rulename{T-value}]{
  \Value\in\Type
}{
  \wte\VEnv\Value\Type
}
\qquad
\inferrule[\rulename{T-eq}]{
  \wte\VEnv{\Expr_1}\Type
  \\
  \wte\VEnv{\Expr_2}\Type
}{
  \wte\VEnv{\Expr_1=\Expr_2}\tbool
}
\qquad
\inferrule[\rulename{T-input}]{
  \wtp{\VEnv,\var:\Type}\Process{\SessionTypeT}{\SessionTypeS}
}{
  \wtp\VEnv{\inp\Channel\var\Type.\Process}{\In\Type.\SessionTypeT}{\SessionTypeS}
}
\\\\
\inferrule[\rulename{T-output}]{
  \wte\VEnv\Expr\Sort
  \\
  \wtp\VEnv\Process\SessionTypeT\SessionTypeS
}{
  \wtp\VEnv{\outp\Channel\Expr.\Process}{\Out\Type.\SessionTypeT}{\SessionTypeS}
}
\qquad
\inferrule[\rulename{T-branch}]{
  \wtp\VEnv{\Process_i}{\SessionType_i}{\SessionTypeS}~{}^{(i=1,2)}
}{
  \wtp\VEnv{ \casep\Channel{\Process_1}{\Process_2}}{\SessionType_1 + \SessionType_2}{\SessionTypeS}
}
\\\\
\inferrule[\rulename{T-select left}]{
  \wtp\VEnv\Process{\SessionType_1}{\SessionTypeS}
}{
  \wtp\VEnv{ \selectp\Channel\selectL.\Process}{\SessionType_1\oplus \SessionType_2}{\SessionTypeS}
}
\quad
\inferrule[\rulename{T-select right}]{
  \wtp\VEnv\Process{\SessionType_2}{\SessionTypeS}
}{
  \wtp\VEnv{ \selectp\Channel\selectR.\Process}{\SessionType_1 \oplus \SessionType_2}{\SessionTypeS}
}
\quad
\inferrule[\rulename{T-idle}]{}{
  \wtplr\VEnv\nullp\End\End
}
\\\\
\inferrule[\rulename{T-conditional}]{
  \wte\VEnv\Expr\tbool
  \\
  \wtplr\VEnv{\Process_i}{\SessionTypeT}{\SessionTypeS}~{}^{(i=1,2)}
}{
  \wtplr\VEnv{\condp\Expr{\Process_1}{\Process_2}}{\SessionTypeT}{\SessionTypeS}
}
\qquad
\inferrule[\rulename{T-parallel}]{
  \wtplr\VEnv\ProcessP{\SessionTypeT}{\SessionTypeT'}
  \\
  \wtplr\VEnv\ProcessQ{\co\SessionTypeT'}{\SessionTypeS}
}{
  \wtplr\VEnv{\ProcessP \pc \ProcessQ}{\SessionTypeT}{\SessionTypeS}
}
\end{array}
\end{math}
}
\end{table}

\begin{theorem}
  \label{tc}
  If $\wtplr{}\Process\End\End$, then $\Process$ is correct.
\end{theorem}
\begin{proof}
  Looking at the typing rules it is clear that $\PP$ can only be
  $\nullp$, or a conditional  or a parallel composition.  The first two case are
  immediate.  In the third case let $\PP$ be
  $\PP_1\pc\ldots\pc\PP_i\pc\ldots\pc\PP_n$, where
  $\PP_1,\ldots,\PP_i,\ldots, \PP_n$ are single-threaded. Then rule
  \rulename{T-parallel} requires 
  \begin{center}$\derlr{}{\PP_1}{\eend}{\T_1},
  \derlr{}{\PP_i}{\overline{\T_{i-1}}}{\T_i}$ for $2\leq i\leq n-1$
  and $\derlr{}{\PP_n}{\overline{\T_{n-1}}}{\eend}$
  \end{center} for some types
  $\T_1,\ldots,\T_{n-1}$.  The proof is by induction on
  $\T_1,\ldots,\T_{n-1}$. The first step coincides with the first
  case. For the induction step we can assume that
  $\PP_1,\ldots,\PP_i,\ldots, \PP_n$ are not conditionals, since otherwise at least one of them could be reduced by rule \rulename{r-cond}.  Notice
  that $\rc$ is the only channel in $\PP_1$ and $\lc$ is the only
  channel in $\PP_n$. Then there must be at least one index $j$ ($1\leq j \leq
  n-1$) such that $\PP_j$ 
  starts  with a communication/selection/branching on channel $\rc$  and $\PP_{j+1}$ starts  with a communication/selection/branching on channel $\lc$.  We only consider the case
  $\T_j=\tsel{\T_\selectL} {\T_\selectR}$, the proofs for the other
  cases being similar. Rules \rulename{T-select left},
  \rulename{T-select right} and \rulename{T-branch} require
  $\PP_j\equiv\selectp\rc\Selector.\ProcessQ$ and
  $\PP_{j+1}\equiv\casep\lc{\ProcessQ_\selectL}{\ProcessQ_\selectR}$. Therefore
  $\PP\red\PP_1\pc\ldots\pc\Q\pc\Q_\Selector\pc\ldots\pc\PP_n$ by
  rules \rulename{r-choice 1} and \rulename{r-context}. This concludes
  the proof, since $\derlr{}{\Q}{\overline{\T_{j-1}}}{\T_\Selector}$,
  $\derlr{}{\Q_\Selector}{\overline{\T_\Selector}}{\T_{j+1}}$.
\end{proof}

To have an isomorphism between two session types $\SessionTypeT$ and
$\SessionTypeS$, we need a process $\AdapterA$ that behaves according
to $\co\SessionTypeT$ on its left interface and according to
$\SessionTypeS$ on its right interface. In this way, the process
``transforms'' $\SessionTypeT$ into $\SessionTypeS$. Symmetrically,
there must be a process $\AdapterB$ that performs the inverse
transformation. Not all of these transformations are isomorphisms,
because we also require that these transformations must not entail any
\emph{loss of information}.
Given a session type $\SessionTypeT$, the simplest process with this
property is the \emphdef{identity} process $\id\SessionTypeT$ defined
below:
\[
\id\End = \nullp
\qquad
\begin{array}{r@{~}c@{~}l}
\id{\Out\Type.\SessionType} & = &
\inp\ChannelL\var\Type.\outp\ChannelR\var.\id\SessionType
\\
\id{\In\Type.\SessionType} & = &
\inp\ChannelR\var\Type.\outp\ChannelL\var.\id\SessionType
\end{array}
\qquad
\begin{array}{r@{~}c@{~}l}
\id{\SessionTypeT \oplus \SessionTypeS} & = &
\casep\ChannelL{\selectp\ChannelR\selectL.\id\SessionTypeT}{\selectp\ChannelR\selectR.\id\SessionTypeS}
\\
\id{\SessionTypeT + \SessionTypeS} & = &
\casep\ChannelR{\selectp\ChannelL\selectL.\id\SessionTypeT}{\selectp\ChannelL\selectR.\id\SessionTypeS}
\end{array}
\]

Notice that
$\wtplr{}{\id\SessionType}{\co\SessionType}\SessionType$. We can now
formalize the notion of session type isomorphism:

\begin{definition}[isomorphism]
\label{def:isom}
We say that the session types $\SessionTypeT$ and $\SessionTypeS$ are
\emphdef{isomorphic}, notation $\SessionType \isom \SessionTypeS$, if
there exist two processes $\AdapterA$ and $\AdapterB$ such that
$\wtplr{}{\AdapterA}{\co\SessionTypeT}{\SessionTypeS}$ and
$\wtplr{}{\AdapterB}{\co\SessionTypeS}{\SessionTypeT}$ and $\AdapterA
\pc \AdapterB \eqp \id\SessionTypeT$ and $\AdapterB \pc \AdapterA \eqp
\id\SessionTypeS$.
\end{definition}

\begin{example}\label{ex1}
  Let $\SessionTypeT \eqdef \Out\tint.\Out\tbool.\End$ and
  $\SessionTypeS \eqdef \Out\tbool.\Out\tint.\End$ and observe that
  $\SessionTypeT$ and $\SessionTypeS$ differ in the order in which
  messages are sent. Then we have $\SessionTypeT \isom
  \SessionTypeS$. Indeed, if we take
\[
\AdapterA \eqdef
  \inp\ChannelL\varX\tint.  \inp\ChannelL\varY\tbool.
  \outp\ChannelR\varY.  \outp\ChannelR\varX.  \nullp
  \text{\quad  and \quad}
  \AdapterB
  \eqdef
  \inp\ChannelL\varX\tbool.\inp\ChannelL\varY\tint.\outp\ChannelR\varY.\outp\ChannelR\varX.\nullp
\]
  we derive $\wtplr{}{\AdapterA}{\co\SessionTypeT}{\SessionTypeS}$ and
  $\wtplr{}{\AdapterB}{\co\SessionTypeS}{\SessionTypeT}$ and moreover
  $\AdapterA \pc \AdapterB \eqp \id\SessionTypeT$ and $\AdapterB \pc
  \AdapterA \eqp \id\SessionTypeS$.
%
%
\eoe
\end{example}

\begin{example}
\label{ex:guess}
Showing that two session types are \emph{not} isomorphic is more
challenging since we must prove that there is no pair of processes
$\AdapterA$ and $\AdapterB$ that turns one into the other without
losing information. We do so reasoning by contradiction. Suppose for
example that $\Out\tint.\End$ and $\End$ are isomorphic. Then, there
must exist $\wtplr{}{\AdapterA}{\In\tint.\End}{\End}$ and
$\wtplr{}{\AdapterB}{\End}{\Out\tint.\End}$.
The adapter $\AdapterB$ is suspicious, since it must send a message of
type $\tint$ on channel $\ChannelR$ without ever receiving such a
message from channel $\ChannelL$. Then, it must be the case that
$\AdapterB$ ``makes up'' such a message, say it is $n$ (observe that
our calculus is deterministic, so $\AdapterB$ will always output the
same integer $n$).
We can now unmask $\AdapterB$ showing a context that distinguishes
$\id{\Out\tint.\End}$ from $\AdapterA \pc \AdapterB$. Consider
\[
  \Context \eqdef 
  \outp\ChannelR{n+1}.\nullp
  \pc
  \hole
  \pc
  \inp\ChannelL\var\tint.
  \cond{\var = n+1}{\nullp}{\outp\ChannelR\False.\nullp}
\]
and observe that $\Context[\id{\Out\tint.\End}]$ is correct whereas
$\Context[\AdapterA \pc \AdapterB]$ is not because
\[
\Context[\AdapterA \pc \AdapterB] \red^*
\nullp \pc \condp{n = n+1}{\nullp}{\outp\ChannelR\False.\nullp}
\red
\nullp \pc \outp\ChannelR\False.\nullp \nred
\]
This means that $\AdapterA \pc \AdapterB \not\eqp
\id{\Out\tint.\End}$, contradicting the hypothesis that $\AdapterA$
and $\AdapterB$ were the witnesses of the isomorphism $\Out\tint.\End
\isom \End$.
\eoe
\end{example}

\begin{example}
  Another interesting pair of non-isomorphic types is given by
  $\SessionTypeT \eqdef \In\tint.\Out\tbool.\End$ and $\SessionTypeS
  \eqdef \Out\tbool.\In\tint.\End$.
  A lossless transformation from $\SessionTypeS$ to $\SessionTypeT$
  can be realized by the process
\[
\AdapterB \eqdef
  \inp\ChannelL\varX\tbool.\inp\ChannelR\varY\tint.\outp\ChannelR\varX.\outp\ChannelL\varY.\nullp\,,
\]
  which reads one message from each interface and forwards it to the
  opposite one.
  The inverse transformation from $\SessionTypeT$ to $\SessionTypeS$
  is unachieavable without loss of information. Such process
  necessarily sends at least one message (of type $\tint$ or of type
  $\tbool$) on one interface \emph{before} it receives the message of
  the same type from the opposite interface. Therefore, just like in
  Example~\ref{ex:guess}, such process must guess the message to send,
  and in most cases such message does not coincide with the one the
  process was supposed to forward.
  \eoe
\end{example}



%% file: axioms.tex
\begin{table}
  \caption{Session type isomorphisms.\strut}\label{tab:axioms}
\framebox[\textwidth]{
\begin{math}
\begin{array}{rr@{~}c@{~}l@{\qquad}rr@{~}c@{~}l}
  \rulename{a1} & \Out\TypeT.\Out\TypeS.\SessionType & \isom & \Out\TypeS.\Out\TypeT.\SessionType &
  \rulename{a2} & \In\TypeT.\In\TypeS.\SessionType & \isom & \In\TypeS.\In\TypeT.\SessionType \\
  \rulename{a3} & \Out\TypeT.(\SessionTypeT \oplus \SessionTypeS) & \isom & \Out\TypeT.\SessionTypeT \oplus \Out\TypeT.\SessionTypeS &
  \rulename{a4} & \In\TypeT.(\SessionTypeT + \SessionTypeS) & \isom &
  \In\TypeT.\SessionTypeT + \In\TypeT.\SessionTypeS \\
  \rulename{a5} & \Out\tunit.\SessionType & \isom & \SessionType &
  \rulename{a6} & \In\tunit.\SessionType& \isom & \SessionType \\
  \rulename{a7} & \Out\tbool.\SessionType& \isom &  \SessionType \oplus \SessionType &
  \rulename{a8} & \In\tbool.\SessionType & \isom &  \SessionType + \SessionType\\
  \rulename{a9} & \SessionTypeT \oplus \SessionTypeS & \isom & \SessionTypeS \oplus \SessionTypeT &
  \rulename{a10} & \SessionTypeT + \SessionTypeS & \isom & \SessionTypeS + \SessionTypeT \\
  \rulename{a11} & (\SessionTypeT_1 \oplus \SessionTypeT_2) \oplus \SessionType_3 & \isom & \SessionTypeT_1 \oplus (\SessionTypeT_2 \oplus \SessionType_3) &
  \rulename{a12} & (\SessionType_1 + \SessionType_2) + \SessionType_3 & \isom &  \SessionType_1 + (\SessionType_2 + \SessionType_3) 
\end{array}
\end{math}
}
\end{table}

Table~\ref{tab:axioms} gathers the session type isomorphisms that we
have identified.
There is a perfect duality between the odd-indexed axioms (about
outputs/selections, on the left) and the even-indexed axioms (about inputs/branchings, on
the right), so we briefly discuss the odd-indexed axioms only.
Axiom~\rulename{a1} is a generalization of the isomorphism discussed
in Example \ref{ex1} and is proved by a similar adapter.
Axiom~\rulename{a3} distributes the \emph{same} output on a
selection. Basically, this means that the moment of selection is
irrelevant with respect to other adjacent output operations.
Axiom~\rulename{a5} shows that sending the unitary value provides no
information and therefore is a superfluous operation.
Axiom~\rulename{a7} shows that sending a boolean value is equivalent
to making a selection, provided that the continuation does not depend
on the particular boolean value that is sent. In general, any data
type with finitely many values can be encoded as possibly nested
choices.
Axiom~\rulename{a9}, corresponding to the commutativity of $\oplus$
wrt $\isom$, shows that the actual label used for making a selection
is irrelevant, only the continuation matters.
Axiom~\rulename{a11}, corresponding to the associativity for $\oplus$
wrt $\isom$, generalizes the irrelevance of labels seen
in~\rulename{a9} to nested selections.
Since $\isom$ is a congruence, the axioms in Table~\ref{tab:axioms}
can also be closed by transitivity and arbitrary session type
contexts.
%

%


%% file: adaptors.tex
\begin{table}\caption{Adapters for type isomorphism.}\label{tab:adapters}
\framebox[\textwidth]{$
\begin{array}{llllllll}
  \Adapter_1    =   \inp{\lc}{x}{\TypeT}.\inp{\lc}{y}{\TypeS}. \procout{\rc}{y}{}\procout{\rc}{x}{}\idA{\T}{\T}  \qquad \qquad
  \AdapterB_1  =   \inp{\lc}{x}{\TypeS}.\inp{\lc}{y}{\TypeT}. \procout{\rc}{y}{}\procout{\rc}{x}{}\idA{\T}{\T} \\
  \Adapter_2    =   \inp{\rc}{x}{t}.\inp{\rc}{y}{s}.\procout{\lc}{y}{}\procout{\lc}{x}{}\idA{\T}{\T}            \qquad\qquad
  \AdapterB_2  =   \inp{\rc}{x}{\TypeS}.\inp{\rc}{y}{\TypeT}.\procout{\lc}{y}{}\procout{\lc}{x}{}\idA{\T}{\T}  \\
  \Adapter_3    =   \inp{\lc}{x}{t}.\procbr{\lc}{\procL{\rc}{}\procout{\rc}{x}{}\idA{\SessionTypeT}{\SessionTypeT}}{\procR{\rc}{}\procout{\rc}{x}{}\idA{\SessionTypeS}{\SessionTypeS}}\\
  \AdapterB_3  =  \procbr{\lc}{\inp{\lc}{x}{\TypeT}.\procout{\rc}{x}{}\procL{\rc}{}\idA{\SessionTypeT}{\SessionTypeT}}{\inp{\lc}{x}{\TypeT}.\procout{\rc}{x}{}\procR{\rc}{}\idA{\SessionTypeS}{\SessionTypeS}}\\
  \Adapter_4  =   \procbr{\rc}{\inp{\rc}{x}{t}.\procout{\lc}{x}{}\procL{\lc}{}\idA{\SessionTypeT}{\SessionTypeT}}{\inp{\rc}{x}{t}.\procout{\lc}{x}{}\procR{\lc}{}\idA{\SessionTypeS}{\SessionTypeS}}\\
   \AdapterB_4  =  \inp{\rc}{x}{t}.\procbr{\rc}{\procL{\lc}{}\procout{\lc}{x}{}\idA{\SessionTypeT}{\SessionTypeT}}{\procR{\lc}{}\procout{\lc}{x}{}\idA{\SessionTypeS}{\SessionTypeS}}\\
  \Adapter_5  =  \inp{\lc}{x}{\unit}.\idA{T}{T}  \qquad \qquad \qquad \qquad \qquad \quad
  \AdapterB_5  =  \procout{\rc}{\Unit}{}\idA{\T}{\T} \\
  \Adapter_6  =  \procout{\lc}{\Unit}{}\idA{\T}{\T}  \qquad \qquad \qquad \qquad \qquad \qquad \quad \;
  \AdapterB_6  =  \inp{\rc}{x}{\unit}.\idA{\T}{\T} \\
 \Adapter_7  = \inp{\lc}{x}{\bool}.\cond{x}{\procL{\rc}{\idA{\T}{\SessionTypeT}}}{\procR{\rc}{\idA{\SessionTypeT}{\SessionTypeS}}} \;\;\;
          \AdapterB_7  =  \procbr{\lc}{\procout{\rc}{\true}{}\idA{\SessionTypeT}{\T}}{\procout{\rc}{\false}{}\idA{\SessionTypeT}{\T}} \\
  \Adapter_8  =  \procbr{\rc}{\procout{\lc}{\true}{\idA{\T}{\SessionTypeT}}}{\procout{\lc}{\false}{\idA{\T}{\SessionTypeT}}}  \;\;
         \AdapterB_8  =  \inp{\rc}{x}{\bool}.\cond{x}{\procL{\lc}{\idA{\SessionTypeT}{\T}}}{\procR{\lc}{\idA{\SessionTypeT}{\T}}} \\
  \Adapter_9    = \procbr{\lc}{\procR{\rc}{\idA{\SessionTypeT}{\SessionTypeT}}}{\procL{\rc}{\idA{\SessionTypeS}{\SessionTypeS}}}  \qquad \qquad \qquad \;
   \AdapterB_9  = \procbr{\lc}{\procR{\rc}{\idA{\SessionTypeS}{\SessionTypeS}}}{\procL{\rc}{\idA{\SessionTypeT}{\SessionTypeT}}} \\
  \Adapter_{10}    = \procbr{\rc}{\procR{\lc}{\idA}{\SessionTypeS}{\SessionTypeS}}{\procL{\lc}{\idA{\SessionTypeT}{\SessionTypeT}}}  \qquad \qquad \quad \;\;\;\;
  \AdapterB_{10}  =   \procbr{\rc}{\procR{\lc}{\idA}{\SessionTypeT}{\SessionTypeT}}{\procL{\lc}{\idA{\SessionTypeS}{\SessionTypeS}}} \\
  \Adapter_{11}  =  \procbr{\lc}{\procbr{\lc}{\procL{\rc}{\idA{\SessionTypeT_1}{\SessionTypeT}}}{\procR{\rc}{}\procL{\rc}{}{\idA{\SessionTypeT_2}{\SessionTypeS}}}}{\procR{\rc}{\procR{\rc}{\idA{\T_3}{\T_3}}}} \\
             \AdapterB_{11}  =  \procbr{\lc}{\procL{\rc}{\procL{\rc}{\idA{\SessionTypeT_1}{\SessionTypeT}}}}{\procbr{\lc}{\procL{\rc}{\procR{\rc}{\idA{\SessionTypeT_2}{\SessionTypeS}}}}{\procR{\rc}{\idA{\T_3}{\T_3}}}}\\
  \Adapter_{12}  =  \procbr{\rc}{\procL{\lc}{\procL{\lc}{\idA{\SessionTypeT_1}{\SessionTypeT}}}}{\procbr{\rc}{\procL{\lc}{\procR{\lc}{\idA{\SessionTypeT_2}{\SessionTypeS}}}}{\procR{\lc}{\idA{\T_3}{\T_3}}}}\\             
             \AdapterB_{12}  =  \procbr{\rc}{\procbr{\rc}{\procL{\lc}{\idA{\SessionTypeT_1}{\SessionTypeT}}}{\procR{\lc}{}\procL{\lc}{}{\idA{\SessionTypeT_2}{\SessionTypeS}}}}{\procR{\lc}{\procR{\lc}{\idA{\T_3}{\T_3}}}}
\end{array}
$}
\end{table}

\begin{table}[t]
\caption{\label{tab:areduction} Symbolic reduction relation.\strut}
\framebox[\textwidth]{
$
\begin{array}{l}
\begin{array}{ll}
\rulename{sr-up 1} \;
 \inp\ChannelL\var\TypeT.\ProcessP
  \pc
  \ProcessQ
  \ared
  \inp\ChannelL\var\TypeT.(\ProcessP
  \pc
  \ProcessQ)
 &
\!\!\!\!\!\!\rulename{sr-up 2} \;
 \ProcessP
  \pc\inp\ChannelR\var\TypeT.
  \ProcessQ
  \ared
  \inp\ChannelR\var\TypeT.(\ProcessP
  \pc
  \ProcessQ)
\\[1pt]
\rulename{sr-up 3} \;
   \outp\ChannelL\vartX.\ProcessP
  \pc
\ProcessQ
  \ared
   \outp\ChannelL\vartX.(\ProcessP
  \pc
  \ProcessQ)
&
\!\!\!\!\!\!\rulename{sr-up 4} \; 
  \ProcessP
  \pc
  \outp\ChannelR\vartX.\ProcessQ
  \ared
   \outp\ChannelR\vartX.(\ProcessP
  \pc
  \ProcessQ)
\\[1pt]
\rulename{sr-up 5}  \;
\casep\ChannelL{\ProcessP_\selectL}{\ProcessP_\selectR}
\pc \ProcessQ \ared \casep\ChannelL{\ProcessP_\selectL\pc \ProcessQ }{\ProcessP_\selectR\pc \ProcessQ }
&
\!\!\!\!\!\!\rulename{sr-up 7}  \;
   \selectp\ChannelL\Selector.\ProcessP\pc\ProcessQ
  \ared
  \selectp\ChannelL\Selector.(\ProcessP \pc \ProcessQ)
\\[1pt]
\rulename{sr-up 6} \;
  \ProcessP
  \pc
  \casep\ChannelR{\ProcessQ_\selectL}{\ProcessQ_\selectR}
  \ared
  \casep\ChannelR{\ProcessP\pc\ProcessQ_\selectL}{\ProcessP\pc\ProcessQ_\selectR}
&
\!\!\!\!\!\!\rulename{sr-up 8} \;
  \ProcessP\pc \selectp\ChannelR\Selector.\ProcessQ
  \ared
  \selectp\ChannelR\Selector.(\ProcessP \pc \ProcessQ) 
\\[1pt]
\end{array}
\\
\begin{array}{l}
\rulename{sr-up  9} \; 
  (\condp\var{\ProcessP_1}{\ProcessP_2})\pc \ProcessQ \ared \condp\var{(\ProcessP_1\pc \ProcessQ)}{(\ProcessP_2\pc \ProcessQ)}
\\[1pt]
\rulename{sr-up  10} \;
 \ProcessP \pc (\condp\var{\ProcessQ_1}{\ProcessQ_2}) \ared \condp\var{(\ProcessP\pc \ProcessQ_1)}{(\ProcessP\pc \ProcessQ_2)} 
 \\[1pt]
\end{array}
\\
\begin{array}{ll}
\rulename{sr-swap 1} \;
  \inp\Channel\var\TypeT.\inp{\co\Channel}\varY\TypeS.\ProcessP
  \ared
  \inp{\co\Channel}\varY\TypeS. \inp\Channel\var\TypeT.\ProcessP
& 
\rulename{sr-swap 2} \;
 \outp\Channel\var. \outp{\co\Channel}{\varY}.\ProcessP
 \ared
  \outp{\co\Channel}{\varY}.\outp\Channel\var.\ProcessP
\end{array}
\\[1pt]
\begin{array}{l}
\rulename{sr-swap 3} \;
 \inp\Channel\var\TypeT. \outp{\co\Channel}\varY.\ProcessP
 \aredLR
  \outp{\co\Channel}\varY.\inp\Channel\var\TypeT.\ProcessP
\quad\var\not=\varY
\\[1pt]
\rulename{sr-swap 4} \; 
 \inp\Channel\var\TypeT. \selectp{\co\Channel}\Selector.\ProcessP
 \aredLR
  \selectp{\co\Channel}\Selector.\inp\Channel\var\TypeT.\ProcessP
\\[1pt]
\rulename{sr-swap 5} \; 
  \outp\Channel\var. \selectp{\co\Channel}\Selector.\ProcessP
 \aredLR
  \selectp{\co\Channel}\Selector. \outp\Channel\var.\ProcessP
\\[1pt]
%
\rulename{sr-swap 6} \;
 \inp\Channel\var\TypeT. \casep{\co\cc}\ProcessP\ProcessQ
 \aredLR
  \casep{\co\cc}{ \inp\Channel\var\TypeT. \ProcessP}{ \inp\Channel\var\TypeT. \ProcessQ}
\\[1pt]
\rulename{sr-swap 7} \;
  \outp\Channel\var.  \casep{\co\cc}\ProcessP\ProcessQ
 \aredLR
  \casep{\co\cc}{\outp\Channel\var. \ProcessP}{\outp\Channel\var. \ProcessQ}
\\[1pt]
\rulename{sr-swap 8} \;
  \casep{\Channel}{\selectp{\co\Channel}\Selector. \ProcessP}{\selectp{\co\Channel}\Selector.\ProcessQ}  
 \aredLR
 \selectp{\co\Channel}\Selector.  \casep{\cc}\ProcessP\ProcessQ
\\[1pt]
\rulename{sr-swap 9} \; 
\selectp{\Channel}\Selector. \selectp{\co\Channel}{\Selector'}. \ProcessP
 \aredLR
  \selectp{\co\Channel}{\Selector'}.\selectp{\Channel}\Selector.\ProcessP
\\[1pt]
\rulename{sr-swap 10} \;
  \casep{\Channel}{\casep{\co\Channel}{\ProcessP_1}{\ProcessQ_1}}
  {\casep{\co\Channel}{\ProcessP_2}{\ProcessQ_2}}  
 \aredLR
 \casep{\co\Channel}{\casep{\Channel}{\ProcessP_1}{\ProcessP_2}}
  {\casep{\Channel}{\ProcessQ_1}{\ProcessQ_2}} 
\\[1pt]
\rulename{sr-cond} \;
{
  \condp\var{\outp{\Channel}\true.\ProcessP}{\outp\Channel\false.\ProcessP} \ared \outp{\Channel}\vartX.\ProcessP
}
\\[1pt]
\end{array} 
\\
\begin{array}{ll}
\rulename{sr-comm 1} \; 
  \outp\ChannelR\vartY.\ProcessP
  \pc
  \inp\ChannelL\var\TypeT.\ProcessQ
  \ared
  \ProcessP
  \pc
  \ProcessQ\subst\varY\var
&
\rulename{sr-comm 2} \; 
  \inp\ChannelR\var\TypeT.\ProcessP
  \pc
  \outp\ChannelL\vartY.\ProcessQ
  \ared
  \ProcessP\subst\varY\var
  \pc
  \ProcessQ
\\[1pt]
\rulename{sr-choice 1} \;
  \selectp\ChannelR\Selector.\ProcessP
  \pc
  \casep\ChannelL{\ProcessQ_\selectL}{\ProcessQ_\selectR}
  \ared
  \ProcessP \pc \ProcessQ_\Selector
&
\rulename{sr-choice 2} \;
  \casep\ChannelR{\ProcessP_\selectL}{\ProcessP_\selectR} \pc \selectp\ChannelL\Selector.\ProcessQ
  \ared
    \ProcessP_\Selector \pc \ProcessQ
\\[1pt]
\rulename{sr-id} \;
  \idA\T\T\pc\idA\T\T \ared \idA\T\T
&
\inferrule[\rulename{sr-contexts}]{
  \ProcessP \ared \ProcessQ}
{
 \ContextS[\ProcessP] \ared \ContextS[\ProcessQ]
}
\end{array}
\end{array}$
}
\end{table}

Table~\ref{tab:adapters} gives all the adapters of the axioms in Table~\ref{tab:axioms}. Then the soundness of the axioms in Table \ref{tab:axioms} amounts to prove:
\begin{equation}\derlr{}{\Adapter_ i}{\overline {\T_i}}{\SessionTypeS_i}\qquad \derlr{}{\AdapterB_ i}{\overline {\SessionTypeS_i}}{\T_i}\label{t}\end{equation}

\begin{equation}\Adapter_ i\pc\AdapterB_ i\eqp\id{\T_i}\qquad\AdapterB_ i\pc\Adapter_ i\eqp\id{\SessionTypeS_i}\label{i}\end{equation}
where $\T_i$ is the l.h.s. and $\SessionTypeS_i$ is the r.h.s. of the axiom \rulename{a$i$} for $1\leq i\leq 12$.

\bigskip

Point \ref{t} can be easily shown by cases on the definitions of $\Adapter_ i$ and $\AdapterB_ i$  taking into account that \[\derlr{}{\id\T}{\overline\T}\T\] for all types $\T$. 

\bigskip

For Point \ref{i} we define a {\em symbolic reduction relation} which preserves equivalence of closed and typed processes (Theorem~\ref{a}). This is enough since we will show that 
  all the parallel compositions of the adapters symbolically reduce to
  the corresponding identities (Theorem~\ref{b}).  
The rules of this relation are given in 
Table  \ref{tab:areduction}, where $\aredLR$ stands for reduction in both directions and 
\emph{symbolic reduction contexts} $\ContextS$ are defined by:
\begin{eqnarray*}
   \ContextS & ::= &  \hole ~~
~~\mid~~ \inp\Channel\var\Type.\ContextS
~~\mid~~ \outp\Channel\Expr.\ContextS
~~\mid~~ \selectp\Channel\Selector.\ContextS 
~~\mid~~\casep\Channel\ContextS\ProcessQ 
~~ \mid~~ \casep\Channel\ProcessP\ContextS \\
& \mid & \condp\Expr\ProcessP\ContextS
~~\mid~~ \condp\Expr\ContextS\ProcessQ
\end{eqnarray*}
We call this a symbolic reduction relation because it also reduces processes with free variables. 
We notice that this reduction applied to two parallel processes:
\begin{enumerate}
\item moves up the communications/selections/branchings on the left channel of the left process and  the communications/selections/branchings on the right channel of the right process and the conditionals,
\item  executes the communications/selections/branchings between the  right channel of the left process and the left channel of the right process when possible,
\item eliminates superfluous identities,
\item swaps communications/selections/branchings on different channels when this  is not forbidden by bound variables.
\end{enumerate}
The more interesting rule is \rulename{sr-cond}, that transforms a conditional in an output. 

\begin{theorem}\label{a}
If $\ProcessP$ is a closed and typed process and $\ProcessP\ared^*\ProcessQ$,  then  $\ProcessP\eqp\ProcessQ$. 
\end{theorem}
\begin{proof} The proof is by induction on the reduction of Table~\ref{tab:areduction} and by cases on the last applied rule.  Notice that the proof for the swap rules is immediate, since these rules can be always reversed.
We  consider some interesting cases, in which we assume $\ProcessR_1\pc\ContextS\pc\ProcessR_2\red^*\ProcessR'_1\pc\hole\pc\ProcessR'_2$ (by extending reduction to contexts in the obvious way) and  that $\subst{\vec\Value}{\vec\varY}$ are the substitutions made on the hole in this reduction.

\medskip

\noindent
\rulename{sr-up 1}
If $\ProcessR_1\pc\ContextS[\inp\ChannelL\var\TypeT.\ProcessP
  \pc
  \ProcessQ]\pc\ProcessR_2$ is correct, then each reduction from $\ProcessR_1\pc\ContextS[\inp\ChannelL\var\TypeT.\ProcessP
  \pc
  \ProcessQ]\pc\ProcessR_2$ to $\nullp$ must be of the shape
  \[\begin{array}{l}\ProcessR_1\pc\ContextS[\inp\ChannelL\var\TypeT.\ProcessP
  \pc
  \ProcessQ]\pc\ProcessR_2\red^*\ProcessR'_1\pc(\inp\ChannelL\var\TypeT.\ProcessP
  \pc
  \ProcessQ)\subst{\vec\Value}{\vec\varY}\pc\ProcessR'_2
  \red^*\\\outp\ChannelR\Expr.\ProcessR\pc\inp\ChannelL\var\TypeT.\ProcessP\subst{\vec\Value}{\vec\varY}
  \pc
  \ProcessQ'\red^*\ProcessR\pc\ProcessP\subst{\vec\Value}{\vec\varY}\subst\Value\var\pc
  \ProcessQ'\red^*\nullp\end{array}\] 
 where   $\ProcessR'_1\red^*\outp\ChannelR\Expr.\ProcessR$  with  $\eval\Expr\Value,$  $\Value \in \Type$, and $\ProcessQ\subst{\vec\Value}{\vec\varY}\pc\ProcessR'_2\red^*\ProcessQ'$. 
  We get
  \[\begin{array}{l}\ProcessR_1\pc\ContextS[\inp\ChannelL\var\TypeT.(\ProcessP
  \pc
  \ProcessQ)]\pc\ProcessR_2 \red^*
  \ProcessR'_1\pc\inp\ChannelL\var\TypeT.(\ProcessP
  \pc
  \ProcessQ)\subst{\vec\Value}{\vec\varY}\pc\ProcessR'_2 
  \red^*\\\outp\ChannelR\Expr.\ProcessR\pc\inp\ChannelL\var\TypeT.(\ProcessP
  \pc
  \ProcessQ)\subst{\vec\Value}{\vec\varY}\pc\ProcessR'_2
\red
\ProcessR\pc\ProcessP\subst{\vec\Value}{\vec\varY}\subst\Value\var\pc\ProcessQ\subst{\vec\Value}{\vec\varY}\pc\ProcessR'_2
  \red^*\\\ProcessR\pc\ProcessP\subst{\vec\Value}{\vec\varY}\subst\Value\var\pc
  \ProcessQ'\red^*\nullp\end{array}\]
  Vice versa if $\ProcessR_1\pc\ContextS[\inp\ChannelL\var\TypeT.(\ProcessP
  \pc
  \ProcessQ)]\pc\ProcessR_2$ is correct, then each reduction from \mbox{$\ProcessR_1\pc\ContextS[\inp\ChannelL\var\TypeT.(\ProcessP
  \pc
  \ProcessQ)]\pc\ProcessR_2$}  to $\nullp$ must be of the shape shown above, and the proof concludes similarly. 
 
 \medskip
 
 \noindent
\rulename{sr-up 7} 
 If  $\ProcessR_1 \pc \ContextS[\selectp\ChannelL\selectL.\ProcessP\pc\ProcessQ] \pc \ProcessR_2$ is correct,  then each reduction from 
  $\ProcessR_1 \pc \ContextS[\selectp\ChannelL\selectL.\ProcessP\pc\ProcessQ] \pc \ProcessR_2$ to $\nullp$ must be of the shape
  \[\begin{array}{l}\ProcessR_1 \pc \ContextS[\selectp\ChannelL\selectL.\ProcessP\pc\ProcessQ] \pc \ProcessR_2\red^*\ProcessR'_1\pc(\selectp\ChannelL\selectL.\ProcessP
  \pc
  \ProcessQ)\subst{\vec\Value}{\vec\varY}\pc\ProcessR'_2\red^*\\
  \procbr{\rc}{\ProcessP_\selectL}{\ProcessP_\selectR}\pc \selectp\ChannelL\selectL.\ProcessP\subst{\vec\Value}{\vec\varY}\pc\ProcessQ'\red^* 
  \ProcessP_\selectL \pc \ProcessP\subst{\vec\Value}{\vec\varY} \pc \ProcessQ'  \red^* \nullp\end{array}
  \]
 where $\ProcessR'_1 \red^* \procbr{\rc}{\ProcessP_\selectL}{\ProcessP_\selectR}$ and $\ProcessQ\subst{\vec\Value}{\vec\varY}\pc\ProcessR'_2\red^*\ProcessQ'$. 
  We get
  \[\begin{array}{l}
  \ProcessR_1 \pc \ContextS[\selectp\ChannelL\selectL.(\ProcessP \pc \ProcessQ)] \pc \ProcessR_2
  \red^*
   \ProcessR'_1 \pc \selectp\ChannelL\selectL.(\ProcessP \pc \ProcessQ)\subst{\vec\Value}{\vec\varY} \pc \ProcessR'_2
  \red^*\\ 
  \procbr{\rc}{\ProcessP_\selectL}{\ProcessP_\selectR} \pc \selectp\ChannelL\selectL.(\ProcessP \pc \ProcessQ)\subst{\vec\Value}{\vec\varY} \pc \ProcessR'_2 \red
  \ProcessP_\selectL \pc \ProcessP\subst{\vec\Value}{\vec\varY} \pc \ProcessQ\subst{\vec\Value}{\vec\varY} \pc \ProcessR'_2 \red^*\\
  \ProcessP_\selectL \pc \ProcessP\subst{\vec\Value}{\vec\varY} \pc \ProcessQ'  \red^* \nullp
  \end{array}
 \]
Vice versa if $\ProcessR_1 \pc \ContextS[\selectp\ChannelL\selectL.(\ProcessP \pc \ProcessQ)] \pc \ProcessR_2$ is correct, then each reduction from \mbox{$\ProcessR_1 \pc \ContextS[\selectp\ChannelL\selectL.(\ProcessP \pc \ProcessQ)] \pc \ProcessR_2$} to $\nullp$ must be of the shape shown above, and the proof concludes similarly. 

\medskip

\noindent\rulename{sr-cond}
If  $\ProcessR_1\pc \ContextS[\condp\var{\outp{\rc}\true.\ProcessP}{\outp\rc\false.\ProcessP}] \pc \ProcessR_2$ is correct, then each reduction from 
 $\ProcessR_1\pc \ContextS[\condp\var{\outp{\rc}\true.\ProcessP}{\outp\rc\false.\ProcessP}] \pc \ProcessR_2$ to $\nullp$ must be of the shape
\[\begin{array}{l}\ProcessR_1\pc \ContextS[\condp\var{\outp{\rc}\true.\ProcessP}{\outp\rc\false.\ProcessP}] \pc \ProcessR_2\red^*\\
  \ProcessR'_1\pc \condp\Value{\outp{\rc}\true.\ProcessP\subst{\vec\Value}{\vec\varY}\subst\Value\var}{\outp\rc\false.\ProcessP\subst{\vec\Value}{\vec\varY}\subst\Value\var} \pc \ProcessR'_2\red^*\\ 
\ProcessR_1' \pc \outp\rc\Value. \ProcessP\subst{\vec\Value}{\vec\varY}\subst\Value\var \pc\ProcessR'_2\red^*  \ProcessR_1' \pc \outp\rc\Value. \ProcessP\subst{\vec\Value}{\vec\varY}\subst\Value\var \pc \inp\lc\varZ\bool.\ProcessR \red^*\\\ProcessR_1' \pc \ProcessP\subst{\vec\Value}{\vec\varY}\subst\Value\var \pc \ProcessR\subst\Value\varZ \red^* \nullp
\end{array}\]
where $\Value \in \bool$ since we start from a typed process and $\ProcessR'_2\red^* \inp\lc\varZ\bool.\ProcessR$. We get
\[
\begin{array}{l}
  \ProcessR_1\pc \ContextS[\outp\rc\var.\ProcessP] \pc \ProcessR_2\red^* 
  \ProcessR_1' \pc \outp\rc\Value.\ProcessP\subst{\vec\Value}{\vec\varY}\subst\Value\var \pc \ProcessR'_2\red^* \\
   \ProcessR_1' \pc \outp\rc\Value.\ProcessP\subst{\vec\Value}{\vec\varY}\subst\Value\var \pc \inp\lc\varZ\bool.\ProcessR \red
   \ProcessR_1' \pc \ProcessP\subst{\vec\Value}{\vec\varY}\subst\Value\var \pc \ProcessR\subst\Value\varZ \red^*\nullp.
\end{array}
\]
Vice versa, if  $\ProcessR_1\pc \ContextS[\outp\rc\var.\ProcessP] \pc \ProcessR_2$ is correct, then each reduction from $\ProcessR_1\pc \ContextS[\outp\rc\var.\ProcessP] \pc \ProcessR_2$ to $\nullp$ must be of the shape shown above with $\Value \in \bool$, 
and the proof 
is similar.
    \end{proof}

\begin{theorem}\label{b}
$\Adapter_ i\pc\AdapterB_ i\ared^*\id{\T_i}$ and $\AdapterB_ i\pc\Adapter_ i\ared^*\id{\SessionTypeS_i}$  for $1\leq i\leq 12$.  
\end{theorem}
\begin{proof} The proof is by cases on $i$. For example\\
$\begin{array}{lcl}
\Adapter_ 1\pc\AdapterB_ 1& \ared^* & \inp{\lc}{x}{\TypeT}.\inp{\lc}{y}{\TypeS}. (\procout{\rc}{y}{}\procout{\rc}{x}{}\idA{\T}{\T}\pc\AdapterB_1) \\
                   & \ared^* & \inp{\lc}{x}{\TypeT}.\inp{\lc}{y}{\TypeS}. (\idA{\T}{\T}\pc\procout{\rc}{x}{}\procout{\rc}{y}{}\idA{\T}{\T})\\
                   & \ared^* & \inp{\lc}{x}{\TypeT}.\inp{\lc}{y}{\TypeS}. \procout{\rc}{x}{}\procout{\rc}{y}{}(\idA{\T}{\T}\pc\idA{\T}{\T})\ared^*\id{\Out\TypeT.\Out\TypeS.\SessionType}\\
 \Adapter_2 \pc \AdapterB_2 & \ared^* &  \inp{\rc}{x}{\TypeT}.\inp{\rc}{y}{\TypeS}.\big(\Adapter_2 \pc \procout{\lc}{y}{}\procout{\lc}{x}{}\idA{\T}{\T}) \\
                     & \ared^* & \inp{\rc}{x}{\TypeT}.\inp{\rc}{y}{\TypeS}.\big(\procout{\lc}{x}{}\procout{\lc}{y}{}\idA{\T}{\T} \pc \idA{\T}{\T}) \\
                     & \ared^* & \inp{\rc}{x}{\TypeT}.\inp{\rc}{y}{\TypeS}.\procout{\lc}{x}{}\procout{\lc}{y}{}(\idA{\T}{\T} \pc \idA{\T}{\T}) \ared^* \id{\In\TypeT.\In\TypeS.\SessionType}\\
   \Adapter_3 \pc \AdapterB_3 
& \ared^* &  \inp{\lc}{x}{\TypeT}.\procbr{\lc}{\procL{\rc}{}\procout{\rc}{x}{}\idA{\SessionTypeT}{\SessionTypeT}\pc \AdapterB_3}{\procR{\rc}{}\procout{\rc}{x}{}\idA{\SessionTypeS}{\SessionTypeS}\pc \AdapterB_3} \\
& \ared^* &  \inp{\lc}{x}{\TypeT}.{\lc}\triangleright\{\procout{\rc}{x}{}\idA{\SessionTypeT}{\SessionTypeT}\pc \inp{\lc}{\x}\TypeT.\procout{\rc}{\x}\procL{\rc}{\idA{\SessionTypeT}{\T}},\\
&&\phantom{\inp{\lc}{x}{\TypeT}.{\lc}\triangleright\{}\procout{\rc}{x}{}\idA{\SessionTypeS}{\SessionTypeS}  \pc  \inp{\lc}{\x}{\TypeT}.\procout{\rc}{\x}\procR{\rc}{\idA{\SessionTypeS}{\S}}\} \\
& \ared^* &  \inp{\lc}{x}{\TypeT}. \procout{\rc}{\x} \procbr{\lc}{\procL{\rc}{\idA{\SessionTypeT}{\T}}}{\procR{\rc}{\idA{\SessionTypeS}{\S}}}  = \id{\Out\TypeT.(\SessionTypeT \oplus \SessionTypeS)}
\end{array}\\
\begin{array}{lcl}
  \Adapter_4 \pc \AdapterB_4 & \ared^* & \inp{\rc}{x}{\TypeT}. ( \Adapter_4 \pc  
     \procbr{\rc}{\procL{\lc}{}\procout{\lc}{x}{}\idA{\SessionTypeT}{\SessionTypeT}}{\procR{\lc}{}\procout{\lc}{x}{}\idA{\SessionTypeS}{\SessionTypeT}} )\\
     & \ared^* & \inp{\rc}{x}{\TypeT}.{\rc}\triangleright \{
      \inp{\rc}{x}{\TypeT}.\procout{\lc}{x}{}\procL{\lc}{}\idA{\SessionTypeT}{\SessionTypeT}  
      \pc \procout{\lc}{x}{}\idA{\SessionTypeT}{\SessionTypeT},\\
      &&\phantom{\inp{\rc}{x}{\TypeT}.{\rc}\triangleright \{}\inp{\rc}{x}{\TypeT}.\procout{\lc}{x}{}\procR{\lc}{}\idA{\SessionTypeS}{\SessionTypeT} \pc \procout{\lc}{x}{}\idA{\SessionTypeS}{\SessionTypeT}\} \\ 
           & \ared^* & \inp{\rc}{x}{\TypeT}.\procbr{\rc}{
      \procout{\lc}{x}{}\procL{\lc}{}\idA{\SessionTypeT}{\SessionTypeT} 
       \pc  \idA{\SessionTypeT}{\SessionTypeT}}{\procout{\lc}{x}{}\procR{\lc}{}\idA{\SessionTypeS}{\SessionTypeT} \pc \idA{\SessionTypeS}{\SessionTypeT}} \\ 
                  & \ared^* & \inp{\rc}{x}{\TypeT}. \procout{\lc}{x}{} \procbr{\rc}{\procL{\lc}{}\idA{\SessionTypeT}{\SessionTypeT} 
       \pc  \idA{\SessionTypeT}{\SessionTypeT}}{\procR{\lc}{}\idA{\SessionTypeS}{\SessionTypeT} \pc \idA{\SessionTypeS}{\SessionTypeT}} \\ 
                         & \ared^* & \inp{\rc}{x}{\TypeT}. \procout{\lc}{x}{} \procbr{\rc}{
      \procL{\lc}{}(\idA{\SessionTypeT}{\SessionTypeT} 
       \pc  \idA{\SessionTypeT}{\SessionTypeT})}{\procR{\lc}{}(\idA{\SessionTypeS}{\SessionTypeT} \pc \idA{\SessionTypeS}{\SessionTypeT}}) \\ 
                & \ared^* & \id{\In\TypeT.(\SessionTypeT + \SessionTypeS)}
\end{array}\\
\begin{array}{lcl}
\Adapter_5 \pc \AdapterB_5 & \ared^* &  \inp{\lc}{x}{\unit}. \procout{\rc}{\Unit}{}(\idA{\T}{\T} \pc \idA{\T}{\T}) \ared^* \id{\Out\tunit.\SessionType}\\                           
\Adapter_6 \pc \AdapterB_6 & \ared^* &  \inp{\rc}{x}{\unit}. \procout{\lc}{\Unit}{} (\idA{\T}{\T}\pc \idA{\T}{\T}) \ared^* \id{\In\tunit.\SessionType} \\
\Adapter_ 7\pc\AdapterB_ 7 & \ared^* & \inp{\lc}{x}{\bool}.\cond{x}{(\procL{\rc}{\idA{\T}{\SessionTypeT}\pc\AdapterB_7)}}{(\procR{\rc}{\idA{\SessionTypeT}{\SessionTypeS}\pc\AdapterB_7})}\\
                   & \ared^* & \inp{\lc}{x}{\bool}.\cond{x}{({\idA{\T}{\SessionTypeT}\pc\outp\rc\true.\idA{\T}{}})}{({\idA{\SessionTypeT}{\SessionTypeS}\pc\outp\rc\false.\idA{\T}{}})}\\
                   & \ared^* & \inp{\lc}{x}{\bool}.\cond{x}{\outp\rc\true.\idA{\T}{}}{\outp\rc\false.\idA{\T}{}} \\
                   & \ared     & \inp{\lc}{\var}{\bool}.\outp\rc \var{}.\idA{\T}{}=\id{\Out\tbool.\SessionType}\\
\end{array}\\
\begin{array}{lcl}
  \Adapter_8 \pc  \AdapterB_8 & \ared^* &  \inp{\rc}{x}{\bool}. \big(\Adapter_8 \pc
                                          \cond{x}{\procL{\lc}{\idA{\SessionTypeT}{\T}}}{\procR{\lc}{\idA{\SessionTypeT}{\T}}}\big) \\ 
                   & \ared^* &   \inp{\rc}{x}{\bool}. \big( \cond{x}{(\Adapter_8 \pc \procL{\lc}{\idA{\SessionTypeT}{\T}})}{(\Adapter_8 \pc {\procR{\lc}{\idA{\SessionTypeT}{\T}}})}\big) \\ 
                   & \ared^* &   \inp{\rc}{x}{\bool}.\big( \cond{x}{\procout{\lc}{\true}{\idA{\T}{\SessionTypeT}}}{\procout{\lc}{\false}{\idA{\T}{\SessionTypeT}}}\big)  \\ 
                   & \ared^* &  \inp{\rc}{x}{\bool}. \procout{\lc}{\x}{\idA{\T}{\SessionTypeT}} = \id{\In\tbool.\SessionType}\\
\end{array}\\
\begin{array}{lcl}
  \Adapter_9  \pc \AdapterB_9 &  \ared^* &\procbr{\lc}{\procR{\rc}{\idA{\SessionTypeT}{\SessionTypeT}}\pc \AdapterB_9}{\procL{\rc}{\idA{\SessionTypeS}{\SessionTypeS}}\pc \AdapterB_9}  \\
                       &   \ared^*&\procbr{\lc}{\idA{\SessionTypeT}{\SessionTypeT}\pc \procL{\rc}{\idA{\SessionTypeT}{\S}}}{\idA{\SessionTypeS}{\SessionTypeS}\pc \procR{\rc}{\idA{\SessionTypeS}{\T}}} \\
                       &   \ared^*&\procbr{\lc}{\procL{\rc}{(\idA{\SessionTypeT}{\T}\pc\idA{\SessionTypeT}{\T})}}{\procR{\rc}{(\idA{\SessionTypeS}{\S}\pc \idA{\SessionTypeS}{\S})}} \ared^* \id{\SessionTypeT \oplus \SessionTypeS}\\
 \Adapter_{10}\pc \AdapterB_{10}  & \ared^* & \procbr{\rc}{\Adapter_{10} \pc \procR{\lc}{\idA}{\SessionTypeT}{\SessionTypeT}}{\Adapter_{10} \pc \procL{\lc}{\idA{\SessionTypeS}{\SessionTypeS}}}  \\
                           & \ared^* & \procbr{\rc}{\procL{\lc}{({\idA{\SessionTypeT}{\SessionTypeT}}\pc \idA{\SessionTypeT}{\T})}}{\procR{\lc}{(\idA{\SessionTypeS}{\SessionTypeS}\pc \idA{\SessionTypeS}{\S})}} \ared^* \id{\SessionTypeT + \SessionTypeS} 
%
\end{array}
$
  
  \end{proof}
  
  Point \ref{i} is a straightforward consequence of Theorems \ref{a} and \ref{b}.

%% file: conclusions.tex
\section{Concluding remarks}
\label{sec:conclusions}

Type isomorphisms have been mainly studied for various
$\lambda$-calculi~\cite{DiCosmo95}. P{\'e}rez et al. \cite{PCPT12} 
interpret intuitionistic linear logic propositions as session types for concurrent processes, which communicate only channels. 
So both their types and their processes differ from ours. In this scenario they explain how type isomorphisms resulting from linear logic equivalences are realized by coercions between interface types of session-based concurrent systems.

The notion of isomorphism for session types investigated in this paper
can be used for automatically adapting behaviors, when their
differences do not entail any loss of information.
Adaptation in general \cite{BCGLV12} is much more permissive than in
our approach, where we require adapters to be invertible. Moreover we
only adapt processes as in \cite{BDPZ12,DP13}, while other works like
\cite{AR12,DGGJL13,CDV14} deal with adaptation of whole
choreographies. 
Our approach shares many similarities
with~\cite{CastagnaGesbertPadovani09,Padovani10} where
\emph{contracts} (as opposed to session types) describe the behavior
of clients and Web services and filters/orchestrators mediate their
interaction. The theory of orchestrators in \cite{Padovani10} allows
not only permutations of subsequent inputs and subsequent outputs, but
also permutations between inputs and outputs if these have no causal
dependencies. The induced morphism is therefore coarser than our isomorphism,
but it may entail some loss of information.

There are some open problems left for future research. The obvious
ones are whether and how our theory extends to recursive and
higher-order session types. Also, we do not know yet whether the set
of axioms in Table~\ref{tab:axioms} is \emph{complete}. The point is
that in the case of arrow, product and sum types or of arrow, intersection,
union types, it is known that the set of isomorphisms is not finitely
axiomatizable \cite{apal:06,DDGT10,CDMZ13}. Despite the fact that
session types incorporate constructs that closely resemble product and
sum types, it may be the case that the particular structure of the
type language allows for a finite axiomatization. A natural question is to what extend our results are a 
consequence of the presence of just two channels in the process language,
or whether they would carry over to calculi with arbitrary channel names. A more interesting research direction is to 
consider this notion of session type isomorphism in
relation to the work on session types and
linear logic \cite{CP10,W12}. 